\documentclass[aps,pra,twocolumn,superscriptaddress]{revtex4-1}

\usepackage{mathrsfs,amsthm,amsmath,graphicx}

\begin{document}


\title{Subspace Stabilization Analysis for Non-Markovian Open Quantum Systems}


\author{Shikun Zhang}
\email[]{daxiayusuozhang@126.com}
\affiliation{School of Automation, Beijing Institute of Technology, Beijing, 100081 ,China}

\author{Kun Liu}
\email[]{kunliubit@bit.edu.cn}
\affiliation{School of Automation, Beijing Institute of Technology, Beijing, 100081 ,China}

\author{Daoyi Dong}
\affiliation{School of Engineering and Information Technology, University of New South Wales, Caberra, ACT, 2600,  Australia}

\author{Xiaoxue Feng}
\email[]{fengxiaoxue@bit.edu.cn}
\affiliation{School of Automation, Beijing Institute of Technology, Beijing, 100081 ,China}

\author{Feng Pan}
\affiliation{School of Automation, Beijing Institute of Technology, Beijing, 100081 ,China}



\date{\today}

\begin{abstract}
Studied in this article is non-Markovian open quantum systems parametrized by Hamiltonian $H$, coupling operator $L$, and memory kernel function $\gamma$, which is a proper candidate for describing the dynamics of various solid-state quantum information processing devices. We look into the subspace stabilization problem of the system from the perspective of dynamical systems and control. The problem translates itself into finding analytic conditions that characterize invariant and attractive subspaces. Necessary and sufficient conditions are found for subspace invariance based on algebraic computations, and sufficient conditions are derived for subspace attractivity by applying a double integral Lyapunov functional. Mathematical proof is given for those conditions and a numerical example is provided to illustrate the theoretical result.
\end{abstract}


\maketitle

\section{Introduction}
Human beings are now in a century when we can not only observe and describe quantum systems, but also alter and control them so as to harness their power unparalleled by classical resources. A promising application lies in Quantum Information Processing (QIP), where exponentially faster computation and provably safer communication are possible to be realised \cite{Nielsen2011Quantum}. In the recent decade, effective QIP devices have been known including silicon photonic crystals \cite{Amir2018Integrating}, trapped ions \cite{Bock2018High}, and superconducting quantum circuits \cite{Shankar2013Autonomously}.

``Quantum information" in the digital world must be represented by, stored in and manipulated through actual physical systems, whose states evolve according to the laws of quantum mechanics and even quantum field theory. Therefore, rigorouly analysing and actively tuning the dynamics of those systems are among the fundamental building blocks of quantum information engineering. This coincides with the basic objective of Systems and Control science, which is to predict the evolution of dynamical systems and make them behave in the way we desire. As a result, quantum control (cybernetics) \cite{Dong2009Quantum, Altafini2012Modeling, Zhang2017Quantum}, born at the intersection of quantum physics, control science and applied mathematics, becomes a useful tool to achieve successful QIP and other quantum engineering applications.

	In this work, we take an in depth look into the subspace stabilization problem which lies in the realm of Systems and Control theory and finds applications in a wide range of QIP problems, e.g., initialization of qubit, generation of entangled states and realization of decoherence-free quantum information. This problem was first studied in \cite{Ticozzi2008Quantum}, where it was analysed in the framework of subspace invariance and attractivity. The authors in \cite{Ticozzi2008Quantum} presented a set of algebraic conditions that characterize invariant and attractive subspaces. Moreover, in \cite{Ticozzi2009Analysis}, sufficient and necessary conditions were derived for invariance and attractivity as opposed to mostly necessary conditions in the previous paper. As subsequent works, the authors in \cite{Ticozzi2010Stabilizing} constructively designed system parameters $(H,L)$ to stabilize generic quantum states, and Ref. \cite{Ticozzi2012Hamiltonian} introduced a computable algorithm to verify those previously proposed conditions and analysed the speed of convergence.

However, the subspace stabilization problem is, up to date, only covered for Lindblad systems \cite{Breuer2002The}. Among the several assumptions that lead to the Lindblad master equation lies the Markovian assumption, which requires that environmental correlations be sufficiently short compared with the system's characteristic time scale. This results in a memoryless, or in other words, Markovian, system where information only flows in one direction. Yet this assumption does not apply to all scenarios. For instance, the modelling of mesoscopic quantum circuits, where field propagation time delay and non-classical input states are considered, often sees the break down of Markovian assumption \cite{Combes2016The}. It seems only natural to extend the analysis of subspace stabilization into the non-Markovian regime.

In the recent decade, non-Markovian quantum systems have attracted increasing interest from the academia. A large amount of work has been done on deriving proper mathematical models, defining and measuring non-Markovianity, and analysing complete positivity; see \cite{Vega2017Dynamics} for an excellent review. However, very few results have addressed the properties of system dynamics given a non-Markovian master equation, which is a topic of major focus for Systems and Control theorists. Therefore, we would like to study the subspace stabilization problem for non-Markovian quantum systems as an investigation of quantum dynamics with memory and for achieving QIP tasks on physical devices with significant non-Markovian effects.

The master equation on which our work bases was derived in \cite{Zhang2013Non} for non-Markovian input-output networks. It applies to atom-like structures in radiation fields, for example, the superconducting ciruit and microwave system. The resulting equation is a time-convolutional one where the derivative of current state depends on all history states and environmental interactions, as opposed to its Markovian (Lindblad) counterpart where only the present state matters. The mathematical object behind time convolutional non-Markovian equations is the Integro-Differential System, see \cite{book:1409764}, from which we have taken a page to help our discussion.

The rest of the paper is organized as follows. In Section II, we introduce the non-Markovian master equation to be studied and define the scope of system parameters to our interest. This is followed by Section III, where the definition of invariant subspaces is given and its \textit{iff} conditions are provided and proved. Section IV presents the definition and sufficient conditions of subspace attractivity, and Section V gives an example of a three level system followed by numerical simulation. The article is concluded by Section VI, which sums up the work and suggests future directions.

\section{Non-Markovian System Model}
In this article, we study non-Markovian open quantum systems described by the following time convolutional master equation, which was derived in \cite{Zhang2013Non} by applying the Born approximation.
\begin{eqnarray}
\dot{\rho}=&-i[H,\rho]+\int_0^t \{\gamma^*(t-\tau)[L\rho(\tau),L_H^{\dagger}(\tau-t)]\nonumber\\
&+\gamma(t-\tau)[L_H(\tau-t),\rho(\tau)L^\dagger] \}d\tau, \label{model}
\end{eqnarray}
where
\begin{equation}
L_H(t)=e^{iHt}Le^{-iHt}.
\end{equation}
There are three parameters in the system model. The Hermitian operator $H$ stands for system Hamiltonian, which generates internal dynamics for the system. Meanwhile, $L$ represents the coupling operator, which describes the interaction interface between the quantum system and its environment. Finally, the memory kernel function $\gamma(t)$ demonstrates the non-Markovianity of the system by weighing the influence of all history system-environment interactions. It is straightforward to verify that this master equation reduces to the well-known and extensively studied Lindblad master equation:
\begin{equation}
\dot{\rho}=-i[H,\rho]+2L\rho L^\dagger-L^{\dagger}L\rho-\rho L^{\dagger}L.
\end{equation}
when $\gamma(t)=\delta(t)$. In this scenario, memoryless kernel function leads to memoryless, or in other words, Markovian, dynamics.

For the sake of simplicity, only real kernel functions are considered in this work. It is also assumed that $\gamma(t)\geq0$, $\gamma(0)\neq0$ and $\gamma\in L^1[0,\infty)$. More restrictions on $\gamma$ may need to be considered to guanrantee complete positivity of the non-Markovian master equation. However, deriving such conditions remains a rather unexplored problem and is beyond the scope of this paper. In fact, complete positivity has been proven in the case of Lorentz spectrum quantum noises (exponentially decaying memory kernels) \cite{Zhang2013Non}, which indicates that completely positive dynamics can be induced by a set of kernel functions that subsumes the exponential family. Therefore, we make a further assumption that $\gamma$ belongs to that set.

Given that the open system evolves under (\ref{model}), its subspace stabilization problem is divided into invariance and attractivity analysis, which will be discussed separately in the following sections.

\section{Subspace Invariance}
This section involves the first half of subspace stabilization problem, subspace invariance. We give a definition of invariant subspaces and present necessary \& sufficient conditions that characterize them. 

Let $\mathcal{H}_I$ be a finite dimensional Hilbert space, and $\mathcal{D}(\mathcal{H}_I)$ be the set of all semi-positive, trace-one, hermitian linear bounded operators on $\mathcal{H}_I$ (density matrices), which forms the state space for quantum system (\ref{model}). The Hilbert space admits the following decomposition,
\begin{equation}
\mathcal{H}_I=\mathcal{H}_S\oplus \mathcal{H}_R,
\end{equation}
where $\mathcal{H}_S=\text{span}\{|\varphi_j^S\rangle\}_{j=0}^m$ and $\mathcal{H}_R=\text{span}\{|\psi_k^R\rangle\}_{k=0}^n$. All basis vectors are orthonormal. According to this subspace decomposition, each operator in (1) has a block matrix representation given this set of bases. We denote those matrices as follows.
\[
H=\left(
     \begin{array}{cc}
       H_S & H_P\\
       H_Q & H_R\\
     \end{array}
   \right),
\quad
L=\left(
     \begin{array}{cc}
       L_S & L_P\\
       L_Q & L_R\\
     \end{array}
   \right)
\]
\[
\rho(t)=\left(
     \begin{array}{cc}
       \rho_S(t) & \rho_P(t)\\
       \rho_Q(t) & \rho_R(t)\\
     \end{array}
   \right),
\]
\[
L_H(t)=\left(
     \begin{array}{cc}
       L_H^S(t) & L_H^P(t)\\
       L_H^Q(t) & L_H^R(t)\\
     \end{array}
   \right).
\]
The hermicity of $H$ and $\rho$ implies that $H_Q=H_P^\dagger$ and $\rho_Q(t)=\rho_P^\dagger(t)$.

We now define what an invariant subspace is. It can be verified that our definition is equivalent to that in \cite{Ticozzi2008Quantum} and \cite{Ticozzi2009Analysis}. However, we simplify the narration in those works by suppressing the notion of quantum subsystems.
\newcounter{definition}
\newcounter{theorem}
\newcounter{corollary}
\newcounter{lemma}
\newtheorem{mydef}[definition]{Definition}
\begin{mydef}[Subspace Invariance] Let the quantum system evolve under (1). $\mathcal{H}_S$ is an invariant subspace if the following condition is satisfied:
\begin{align*}
&\mbox{if}\qquad\rho(0)=\left(
     \begin{array}{cc}
       \rho_S^0 & 0\\
       0 & 0\\
     \end{array}
   \right),\quad \forall \rho_S^0\in \mathcal{D}(\mathcal{H}_S),\\  
&\mbox{then}\qquad\rho(t)=\left(
     \begin{array}{cc}
       \rho_S(t) & 0\\
       0 & 0\\
     \end{array}
   \right),\quad \forall t\geq 0. 
\end{align*}
\end{mydef}
The following theorem completely characterizes invariant subspaces.
\newtheorem{mythm}[theorem]{Theorem}
\begin{mythm}[Subspace Invariance]
The following conditions (i),(ii),(iii) are necessary and sufficient for $\mathcal{H}_S$ to be an invariant subspace.

(i)
\[ 
H=\left(
     \begin{array}{cc}
       H_S & 0\\
       0 & H_R\\
     \end{array}
   \right);\]

(ii)
\[L=\left(
     \begin{array}{cc}
       L_S & L_P\\
       0 & L_R\\
     \end{array}
   \right);\]

(iii) Denote by $\rho_S(t;\rho_S^0)$ the trajectory, with initial value $\rho_S^0$, which satisfies the following integro-differential equation.
\begin{multline}
\dot{\rho}_S=-i[H_S,\rho_S]+\\
\int_0^t \gamma^*(t-\tau)[L_S\rho_S(\tau),L_H^{S\dagger}(\tau-t)]+\mbox{h.c.}d\tau,
\end{multline}
where
\begin{equation}
L_H^S(t)=e^{iH_St}L_Se^{-iH_St}.
\end{equation}
Then, $\forall \rho_S^0 \in \mathcal{D}(\mathcal{H}_S)$,
\begin{equation}
\int_0^t \gamma(t-\tau)\rho_S(\tau;\rho_S^0)L_S^{\dagger}L_H^P(\tau-t)d\tau=0.
\end{equation}
\end{mythm}

\begin{proof}

\textbf{Necessity}.

Suppose $\mathcal{H}_S$ is an invariant subspace, then we have the following relationship according to Definition 1:
\begin{align*}
&\rho(t)=\left(
     \begin{array}{cc}
       \rho_S(t;\rho_S^0) & 0\\
       0 & 0\\
     \end{array}
   \right),\quad \forall t \geq 0, \quad \forall\rho_S^0 \in \mathcal{D}(\mathcal{H}_S);\\
&\dot{\rho}(t)=\left(
     \begin{array}{cc}
       \dot{\rho_S}(t;\rho_S^0) & 0\\
       0 & 0\\
     \end{array}
   \right)=\left(
     \begin{array}{cc}
       S(t) & P(t)\\
       Q(t) & R(t)\\
     \end{array}
   \right),\quad \forall t \geq 0.
\end{align*}
Hermitity of the state density matrix and its derivative imply that $Q(t)=P^{\dagger}(t)$.
We proceed to compute explicitly the S, P and R blocks.
\begin{multline}
S(t)=-i[H_S,\rho_S]\\
+\int_0^t \{ \gamma^*(t-\tau)[L_S\rho_S(\tau),L_H^{S\dagger}(\tau-t)]\\
-L_H^{Q\dagger}(\tau-t)L_Q\rho_S(\tau)\}+\mbox{h.c.}d\tau,
\end{multline}
\begin{multline}
P(t)=i\rho_SH_P+\int_0^t \gamma^*(t-\tau)L_S \rho_S(\tau)L_H^{Q\dagger}(\tau-t)\\
+\gamma(t-\tau)[L_H^S(\tau-t)\rho_S(\tau)L_Q^{\dagger}\\
-\rho_S(\tau)(L_S^{\dagger}L_H^P(\tau-t)+L_Q^{\dagger}L_H^R(\tau-t))] d\tau,
\end{multline}
\begin{multline}
R(t)=\int_0^t\gamma^*(t-\tau)L_Q\rho_S(\tau)L_H^{Q\dagger}(\tau-t)+\mbox{h.c.}d\tau.
\end{multline}
Since $R(t)\equiv 0$, then $\dot{R}(t)\equiv 0$, and $\dot{R}(0)=0$. Changing the integration variable yields:
\begin{equation}
R(t)=\int_0^t \gamma^*(\tau)L_Q\rho_S(t-\tau)L_H^Q(-\tau)+\mbox{h.c.}d\tau,
\end{equation}
\begin{multline}
\dot{R}(t)=\int_0^t \gamma^*(\tau)L_Q\partial_t\rho_S(t-\tau)L_H^Q(-\tau)+\mbox{h.c.}d\tau+\\
\gamma^*(t)L_Q\rho_S^0L_H^{Q\dagger}(-t)+\mbox{h.c.},
\end{multline}
\begin{equation}
\dot{R}(0)=(\gamma^*(0)+\gamma(0))L_Q\rho_S^0L_Q^\dagger=0, \quad \forall \rho_S^0 \in \mathcal{D}(\mathcal{H}_S).
\end{equation}
Therefore, $L_Q=0$. The S and P blocks are thus reduced to:
\begin{multline}
S(t)=-i[H_S,\rho_S]\\
+\int_0^t \gamma^*(t-\tau)[L_S\rho_S(\tau),L_H^{S\dagger}(\tau-t)]+\mbox{h.c.}d\tau,
\end{multline}
\begin{multline}
P(t)=i\rho_SH_P+\int_0^t \gamma^*(t-\tau)L_S \rho_S(\tau)L_H^{Q\dagger}(\tau-t)\\
-\gamma(t-\tau)\rho_S(\tau)L_S^\dagger L_H^P(\tau-t)d\tau.
\end{multline}
Moreover, since $P(0)=i\rho_S^0H_P=0$, the arbitrariness of $\rho_S^0$ indicates that $H_P=0$. It follows that $H$ must have a block diagonal structure, thus leading to the explicit form of $L_H(t)$:
\begin{equation}
L_H(t)=\left(
     \begin{array}{cc}
       e^{iH_St}L_Se^{-iH_St} &e^{iH_St} L_Pe^{-iH_Rt}\\
       0 & e^{iH_Rt}L_Re^{-iH_Rt}\\
     \end{array}
   \right).
\end{equation}
This structure implies that $L_H^Q(t)=0$, which further reduces the P block to:
\begin{equation}
P(t)=-\int_0^t \gamma(t-\tau)\rho_S(\tau,\rho_S^0)L_S^{\dagger}L_H^P(\tau-t)d\tau \equiv 0.
\end{equation}
Necessity is thus proved.

\textbf{Sufficiency}.

Suppose that conditions \textit{(i)}, \textit{(ii)} and \textit{(iii)} are satisfied. Direct computation yields the following  integro-differential equations for sub-blocks of the state density matrix:
\begin{multline}
\dot{\rho}_S(t)=-i[H_S,\rho_S]\\
+\int_0^t \gamma^{*}(t-\tau)\{[L_S\rho_S(\tau)+L_P\rho_P^{\dagger}(\tau),L_H^{S\dagger}(\tau-t)]\\
+(L_S\rho_P(\tau)+L_P\rho_R(\tau))L_H^{P\dagger}(\tau-t)\}+\mbox{h.c.}d\tau,
\end{multline}
\begin{multline}
\dot{\rho}_P(t)=-i(H_S\rho_P-\rho_PH_R)\\
+\int_0^t\gamma^{*}(t-\tau)[(L_S\rho_P(\tau)+L_P\rho_R(\tau))L_H^{R\dagger}(\tau-t)\\
-L_H^{S\dagger}(\tau-t)(L_S\rho_P(\tau)+L_P\rho_R(\tau))]\\
+\gamma(t-\tau)[(L_H^S(\tau-t)\rho_P(\tau)+L_H^P(\tau-t)\rho_R(\tau))L_R^{\dagger}\\
-(\rho_S(\tau)L_S^{\dagger}+\rho_R(\tau)L_P^{\dagger})L_H^P(\tau-t)-\rho_R(\tau)L_R^{\dagger}L_H^R(\tau-t)]d\tau,
\end{multline}
\begin{multline}
\dot{\rho}_R(t)=-i[H_R,\rho_R]\\
+\int_0^t \gamma^{*}(t-\tau)\{[L_R\rho_R(\tau),L_H^{R\dagger}(\tau-t)]\\
-L_H^{P\dagger}(\tau-t)(L_S\rho_P(\tau)+L_P\rho_R(\tau))\}+\mbox{h.c.}d\tau.
\end{multline}
It suffices to verify that $\rho(t;\rho_S^0)$, $\rho_P(t)\equiv0$, and $\rho_R(t)\equiv0$ are solutions of (18), (19), and (20). It is clear that (18) and (20) are satisfied, while (19) leads to:
\begin{equation}
-\int_0^t \gamma(t-\tau)\rho_S(\tau;\rho_S^0)L_S^{\dagger}L_H^P(\tau-t)d\tau=0
\end{equation}
which is satisfied because of condition \textit{(iii)}. This completes the proof of sufficiency.
\end{proof}
Although the conditions given in Theorem 1 are necessary and sufficient, condition \textit{(iii)} may be difficult to verify for systems with high dimensions. Therefore, some useful necessary (not sufficient) and sufficient (not necessary) conditions are provided.
\newtheorem{mycol}[corollary]{Corollary}
\begin{mycol}
Consider the following conditions (iv) and (v):

(iv) $L_S^{\dagger}L_P=0$;

(v)  $[L_S^{\dagger},H_S]=0$.
\begin{flushleft}
(iv) is necessary for $\mathcal{H}_S$ to be invariant. (i), (ii), (iv) and (v) are sufficient for subspace invariance.
\end{flushleft}
\end{mycol}
\begin{proof}
We begin by showing that \textit{(iv)} is necessary.

Calculating the derivative of $P(t)$ at $t=0$ yields:
\begin{equation*}
\dot{P}(0)=-\gamma(0)\rho_S^0 L_S^{\dagger}L_P=0,
\end{equation*}
which holds for arbitrary $\rho_S^0$. This implies that $L_S^{\dagger}L_P=0$.

For the sufficiency of \textit{(i)}, \textit{(ii)}, \textit{(iv)} and \textit{(v)}, we prove that \textit{(iv)} and \textit{(v)} leads to \textit{(iii)}.

This is clear since:
\begin{align*}
&L_S^{\dagger}L_H^P(\tau-t)\\
&=L_S^{\dagger}e^{iH_S(\tau-t)}L_Pe^{-iH_S(\tau-t)}\\
&=e^{iH_S(\tau-t)}L_S^{\dagger}L_Pe^{-iH_S(\tau-t)}\\
&=0.
\end{align*}
Thus ends the proof of Corollary 1.
\end{proof}

After defining and characterizing invariant subspaces, it can be seen that each of them determines an ``invariant set" in $\mathcal{D}(\mathcal{H}_I)$, which is the set of density matrices that are "compressed" within the top left S block. If the initial state locates in that set, all future states will remain in it as long as the system evolves under (\ref{model}). Invariant subspaces thus correspond to preserved quantum information. Moreover, they pave the way for subspace attractivity, which will be discussed in the next section.

\section{Subspace Attractivity}
Building on the analysis of subspace invariance in the previous section, we proceed to define and characterize attractive subspaces. It can also be checked that this is equivalent to the definition in \cite{Ticozzi2008Quantum}.
\newtheorem{mydef1}[definition]{Definition}
\begin{mydef1}[Subspace Attractivity]

Let $\rho(t)$ evolve under (1). If 
\[
{\lim_{t\to +\infty}}(\rho(t)-\left(
     \begin{array}{cc}
       \rho_S(t) & 0\\
       0 & 0\\
     \end{array}
   \right))=0
\]
for all initial states in $\mathcal{D}(\mathcal{H}_I)$, and $\mathcal{H}_S$ is invariant, then $\mathcal{H}_S$ is said to be an attractive subspace.
\end{mydef1}

It is straightforward from this definition that an attractive subspace $\mathcal{H}_S$ is related to an invariant and  attractive set of density matrices.

Suppose $[L_S^{\dagger},H_S]=0$. Then (20) reduces to the following equation considering real $\gamma$ functions.
\begin{multline}
\dot{\rho}_R(t)=-i[H_R,\rho_R]\\
+\int_0^t \gamma(t-\tau)\{[L_R\rho_R(\tau),L_H^{R\dagger}(\tau-t)]\\
+[L_H^{R}(\tau-t),\rho_R(\tau)L_R^\dagger]\}d\tau\\
+\int_0^t \gamma(t-\tau)(-L_H^{P\dagger}(\tau-t)L_P\rho_R(\tau)\\
-\rho_R(\tau)L_P^{\dagger}L_H^{P}(\tau-t))d\tau.
\end{multline}
This implies that the evolution of $\rho_R$ is independent, as opposed to (20), where it also relies on $\rho_P$.

We cast (22) into superoperator form:
\begin{equation}
\dot{\rho}_R=\mathcal{A}\rho_R+\int_0^t \mathcal{B}(t-\tau)\rho_R(\tau)d\tau+\int_0^t \mathcal{K}(t-\tau)\rho_R(\tau)d\tau,
\end{equation}
where
\begin{align*}
&\mathcal{A}[\cdot]=-i[H_R,\cdot],\\
&\mathcal{B}(t)[\cdot]=\gamma(t)\{[L_R\cdot,L_H^{R\dagger}(-t)]+[L_H^{R}(-t),\cdot L_R^\dagger]\},\\
&\mathcal{K}(t)[\cdot]=-\gamma(t)(L_H^{P\dagger}(-t)L_P\cdot+\cdot L_P^\dagger L_H^{P}(-t)).
\end{align*}
Before presenting the main theorem of this section, we shall first prove a useful lemma.
\newtheorem{mylem}[lemma]{Lemma}
\begin{mylem}
Let $f(t)$ be a continuously differentiable function on $[0,\infty)$, and $f(t)\geq 0$. If $\dot{f}(t)\leq \phi(t)$, where $\phi(t)\geq 0$ and $\phi \in L^1 [0,\infty)$, then $f(t)$ must have a finite limit when $t$ tends to infinity.
\end{mylem}
\begin{proof}
We first prove that the statement is correct when $\dot{f}(t)$ has a finite number of zero points.

Let $t_{\mbox{\footnotesize{max}}}$ be the largest zero point. On $(t_{\mbox{\footnotesize{max}}},\infty)$, $\dot{f}(t)$ must either remain negative or positive. If it remains negative, then $f(t)$ must have a limit since it is descending and lower bounded by 0 on $(t_{\mbox{\footnotesize{max}}},\infty)$. If it remains positive, consider the following inequalities.
\begin{eqnarray}
f(t)&=f(0)+\int_0^t \dot{f}(s)ds\nonumber\\
    &\leq f(0)+\int_0^t \phi(s)ds\nonumber\\
    &\leq f(0)+\int_0^{\infty}\phi(s)ds.\nonumber
\end{eqnarray}
Because $\phi \in L^1 [0,\infty)$, $f(t)$ is upper bounded. It thus has a limit since it is increasing on $(t_{\mbox{\footnotesize{max}}},\infty)$.

We then proceed to consider the case where $\dot{f}(t)$ has an infinite number of zero points. The statement can be proved by contradiction. Suppose that $f(t)$ has no limits. Then we have:
\[
U=\varlimsup_{t\to +\infty} f(t)>\varliminf_{t\to +\infty} f(t)=L.
\]
Denote the sequence of peaks by $\{u_n\}_{n=1}^\infty$, and the sequence of valleys by $\{l_n\}_{n=1}^\infty$. The definition of limit superior and limit inferior implies that:
\[
U=\varlimsup_{t\to +\infty} f(t)=\varlimsup_{n\to +\infty}u_n,
\]
\[
L=\varliminf_{t\to +\infty} f(t)=\varliminf_{n\to +\infty}l_n.
\]
The equivalent definition of superior and inferior limits indicates that there exist
\[
\{u_{n_k}\}_{k=1}^\infty \subset \{u_n\}_{n=1}^\infty, \quad \lim_{k\to +\infty}u_{n_k}=U;
\]
\[
\{l_{n_k}\}_{k=1}^\infty \subset \{l_n\}_{n=1}^\infty, \quad \lim_{k\to +\infty}l_{n_k}=L.
\]
We also have that $\forall t,s>0$, $t\geq s$,
\[
f(t)-f(s)=\int_s^t \dot{f}(\tau)d\tau \leq \int_s^t\phi(\tau)d\tau.
\]

Since $\phi \in L^1 [0,\infty)$, $f(t)-f(s)$ tends to 0 when $t$ and $s$ tend to infinity. However, if we pick $\{u_{n_{k_i}}\}_{i=1}^\infty \subset \{u_{n_k}\}_{k=1}^\infty$, s.t. $l_{n_i}\leq u_{n_{k_i}}$, $\forall i\in N^+$, we have:
\[
\lim_{i\to +\infty}(f(u_{n_{k_i}})-f(l_{n_i}))=U-L>0.
\]
This results in a contradiction. Therefore, $f(t)$ must have a limit.
\end{proof}

We are now in the position to present the main result on subspace attractivity.
\newtheorem{mythm1}[theorem]{Theorem}
\begin{mythm1}[Subspace Attractivity]
If $[L_S^\dagger,H_S]=0$, and the matrix
\[
-2\gamma(0)L_P^{\dagger}L_P+\int_t^\infty \Vert \Omega(\tau,t)\Vert d\tau \cdot I
\]
is negative definite for all $t\geq 0$, then $\mathcal{H}_S$ is an attractive subspace. $\Omega(\cdot,\cdot)$ is a two-variable superoperator expressed as:
\begin{multline}
\Omega(t,s)=\mathcal{K}(t-s)-\partial_s\mathcal{K}(t-s)-\mathcal{K}(t-s)\mathcal{A}\\
-\int_s^t\mathcal{K}(t-u)(\mathcal{K}(u-s)+\mathcal{B}(u-s))du.
\end{multline}
\end{mythm1}
\begin{proof}
Since $\rho(t)$ is always positive, it suffices to show that the origin is asymptotically stable (all solutions tend to 0) for (22) and (23). The Variation of Parameters technique of integro-differential equations; see \cite{book:1409764}, allow us to swap (23) into an equivalent equation:
\begin{equation}
\dot{\sigma}=\mathcal{N}\sigma+\int_0^t \mathcal{L}(t,s)\sigma(s)ds+\mathcal{K}(t)\sigma_0,
\end{equation}
where
\[
\mathcal{N}=\mathcal{A}-\mathcal{K}(0),
\]
and
\[
\mathcal{L}=\Omega(t,s)+\mathcal{B}(t,s).
\]
Consider the following Lyapunov functional:
\begin{equation}
V(t,\sigma(\cdot))=tr(\sigma)+\int_0^t\int_t^\infty \Vert\Omega(\tau,s)\Vert d\tau tr(\sigma(s))ds,
\end{equation}
where $\Vert \cdot \Vert$ denotes the norm of superoperators on the Banach space of all Hermitian matrices. Taking the derivative of this functional w.r.t $t$ yields:
\begin{multline}
\dot{V}(t,\sigma(\cdot))=\text{tr}((\mathcal{A}-\mathcal{K}(0))\sigma)+\text{tr}(\int_0^t (\Omega)(t,s)\sigma(s)ds)\\+\text{tr}(\int_0^t (\mathcal{B})(t,s)\sigma(s)ds)
+\text{tr}(\mathcal{K}(t)\sigma_0)\\-\int_0^t \Vert \Omega(t,s) \Vert \text{tr}(\sigma(s))ds+
\int_t^\infty \Vert \Omega(\tau,t) \Vert d\tau \text{tr}(\sigma).
\end{multline}
Using the fact that $\text{tr}(\mathcal{A}[\cdot])=0$ and $\text{tr}(\mathcal{B}(t,s)[\cdot])=0$, and applying the norm inequality we obtain:
\begin{multline*}
\dot{V}(t,\sigma(\cdot)) \leq -\text{tr}(\mathcal{K}(0)\sigma)+\text{tr}(\int_0^t\Vert \Omega(t,s) \Vert \sigma(s)ds )\\
+\text{tr}(\mathcal{K}(t)\sigma_0)
-\int_0^t \Vert \Omega(t,s) \Vert \text{tr}(\sigma(s))ds  \\
+\int_t^\infty \Vert \Omega(\tau,t) \Vert d\tau \text{tr}(\sigma)\\
=\text{tr}([-2\gamma(0)L_P^{\dagger}L_P+\int_t^\infty \Vert \Omega(\tau,t) \Vert d\tau \cdot I]\sigma)+\text{tr}(\mathcal{K}(t)\sigma_0).\nonumber
\end{multline*}
The negative definiteness of \[-2\gamma(0)L_P^{\dagger}L_P+\int_t^\infty \Vert \Omega(\tau,t)\Vert d\tau \cdot I\] implies that $\dot{V}(t,\sigma(\cdot))\leq \text{tr}(\mathcal{K}(t)\sigma_0)$, which is a scalar function in $L^1[0,\infty)$ because $\gamma \in L^1[0,\infty)$ and all other time dependent terms are oscillatory and bounded. 

Therefore, by applying Lemma 1, we know that the Lyapunov functional (26) must have a finite limit when $t$ tends to infinity. The natural boundedness of density matrices implies that the second derivative of $V$ w.r.t $t$ is also bounded. The Barbalat's lemma thus tells us that $\dot{V}$ tends to zero. This leads to the fact that
\[
\lim_{t\to +\infty}\text{tr}([-2\gamma(0)L_P^{\dagger}L_P+\int_t^\infty \Vert \Omega(\tau,t) \Vert d\tau \cdot I]\sigma)=0
\]
because $\text{tr}(\mathcal{K}(t)\sigma_0)$ tends to 0. The negative definiteness again says that $\sigma(t) \to 0$, which completes the proof. 
\end{proof}
\newcounter{remark}
\newtheorem{myrem}[remark]{Remark}
\begin{myrem}
In an attractive subspace, the invariant set determined by the subspace is autonomously stabilized for all initial states. This is interesting for QIP applications, where quantum information may be manipulated and free from decoherence. If the attractive subspace is only one-dimensional, then the set shrinks to a single pure state that spans the subspace. Tasks such as qubit initialization, cooling and entanglement generation can be realized if we choose a proper subspace decomposition.
\end{myrem}

\begin{figure}
\includegraphics[width=0.5\textwidth]{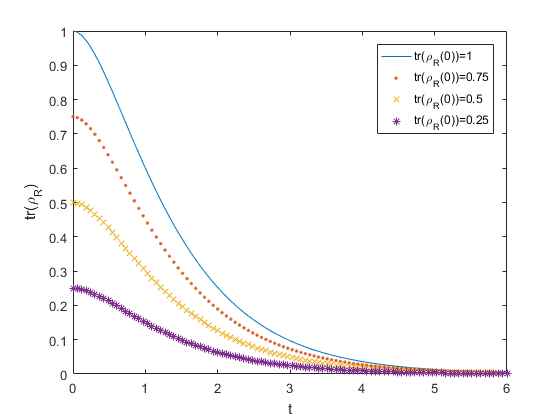}
\caption{Four different initial values for $\text{tr}(\rho_R)$: 1, 0.75, 0.5, 0.25 are chosen. Simulation results show that  $\text{tr}(\rho_R)$ vanishes as time elapses, demonstrating subspace attractivity. }\label{section}
 \end{figure}

\section{Numerical Example and Simulation}

In this section, an example with numerical simulation is presented to illustrate the results.

Consider a three-level system with the following parameters, where we have set $\hbar=1$. The kernel function is $\gamma(t)=e^{-3t}$ and
\[
H=\left(
\begin{array}{ccc}
1/2 & 0 & 0\\
0   &-1/2& 0\\
0 & 0 & -1/2\\ 
\end{array}
\right),\quad
L=\left(
\begin{array}{ccc}
1 & 0 & 0\\
0 & 0 & 1\\
0 & 0 & 0\\ 
\end{array}
\right).
\]

The S block corresponds to the 2$\times$2 block on the top left. The corresponding subspace $\mathcal{H}_S$ is thus 2-dimensional. Its attractivity will be demonstrated via simulation. It can be verified directly that the matrices satisfy sufficient conditions for invariance proposed in Section III. Direct computation yields:

\begin{align*}
&-2\gamma(0)L_P^{\dagger}L_P+\int_t^\infty \Vert \Omega(\tau,t) \Vert d\tau\\
&\leq -2+\int_0^\infty e^{-3u}\vert 4-4u \vert du\\
&\leq -\frac{2}{9}\\
& <0.
\end{align*}

Therefore, sufficient conditions for attractivity are also met. We plot $\text{tr}(\rho_R)$ w.r.t time in FIG.1, choosing 4 different initial values. FIG.1 shows that they converge to 0, meaning that $\mathcal{H}_S$ is attractive.

%

\section{Conclusion and Future Work}
We have extended the analysis of subspace invariance and attractivity to a class of non-Markovian quantum systems. By doing so, we attempt to reach deeper than only to model non-Markovian systems: we have also set foot on investigating their asymptotic dynamical properties, which is among the first few attempts in literature to our knowledge. In future works, other non-Markovian models with potential QIP applications will be investigated. It is also worthwhile investigating whether non-Markovian quantum systems may have other undefined dynamical properties compared with their Markovian counterpart, which only makes future studies much more intriguing.


%



\begin{acknowledgments}
This work has been supported by National Natural Science Foundation (NNSF) of China under Grant 61603040.
\end{acknowledgments}

\nocite{Shabani2005Completely}
\nocite{Pechukas1994Reduced}
\nocite{Pan2016Stabilizing}
\nocite{Vacchini2016Generalized}
\nocite{Xue2012Decoherence}
\nocite{Vacchini2013Non}
\bibliography{zsk1}

\end{document}